\documentclass{article}
\usepackage{amsmath,amssymb,amsfonts,amsthm}
\allowdisplaybreaks
\usepackage{paralist}
\usepackage{booktabs}
\usepackage[round]{natbib}
\usepackage{graphicx}
\usepackage{url}
\usepackage{algorithm} 
\usepackage{algpseudocode} 
\usepackage{xcolor}

\title{What makes you unique?}
\author{

    Benjamin B. Seiler \\ Stanford University
    \and
    Masayoshi Mase\\Hitachi, Ltd.
    \and
    Art B. Owen \\ Stanford University
     
}

\date{October 2022}

\renewcommand{\le}{\leqslant}
\renewcommand{\ge}{\geqslant}
\renewcommand{\emptyset}{{\varnothing}}

\newtheorem{proposition}{Proposition}

\newcommand{\giv}{\!\mid\!}
\newcommand{\kld}{\,\Vert\,} % Kullback-Leibler divergence

\newcommand{\real}{\mathbb{R}}

\newcommand{\bsx}{\boldsymbol{x}}

\newcommand{\bsz}{\boldsymbol{z}}

\newcommand{\val}{\mathrm{val}}

\newcommand{\phz}{\phantom{0}}

\newcommand{\cx}{\mathcal{X}}
\newcommand{\ent}{\mathcal{H}}

\def\mca#1{\multicolumn{1}{c}{#1}}

\begin{document}
\maketitle
\begin{abstract}
This paper proposes a uniqueness Shapley measure
to compare the extent to which different 
variables are able to identify a subject.  
Revealing the value of a variable on subject
$t$ shrinks the set of possible subjects that $t$ could be.
The extent of the shrinkage depends on which other
variables have also been revealed.  We use
Shapley value to combine all of 
the reductions in log cardinality due to 
revealing a variable
after some subset
of the other variables has been revealed.
This uniqueness Shapley measure can be
aggregated over subjects where it becomes 
a weighted sum of conditional entropies. 
Aggregation over subsets of subjects can address questions like how identifying is age for people of a given zip code.  Such aggregates have a
corresponding expression in terms of cross entropies.
We use uniqueness Shapley to investigate
the differential effects of revealing
variables from the North Carolina voter
registration rolls and in identifying anomalous solar flares.  An enormous speedup (approaching 2000 fold in one example)
is obtained by using the all dimension
trees of \cite{moor:lee:1998} to store the
cardinalities we need.
\end{abstract}

\section{Introduction}

An individual data point, such as one representing a person,
can often be identified by specifying even a small
subset of its variables.
For instance, a large fraction of US residents are uniquely identified
by just their date of birth, zip code, and gender
\citep{swee:2000,goll:2006}.
The website \url{https://amiunique.org/} will examine
some signature variables in your browser and report whether
you are uniquely identified among the millions of
participants.  See \cite{gomez2018hiding} for a description.

Variables are not equally powerful for the purpose of
identifying an individual, and the variables that
provide the most information for one person might not
be very informative for another.
It can also happen that two or more variables specified
together can be much more identifying
than we might surmise given their individual strengths
for identifying people. The joint specification can
also be less informative due to associations between
variables such as near duplicates.

In this paper we propose a way to measure how important a variable is for identifying one specific subject in a set of data.  Our definition of importance is based on Shapley value from
economic game theory \citep{shap:1953}.   An important variable is one that, when revealed for a subject, greatly reduces the number of subjects who could match it.  This measure takes account of which other variables might also have been revealed, and so it depends on the full joint distribution
in the data set not just one or a few marginal distributions of the data.  It does not assume that any set of variables necessarily provides a unique identification of the subject, as might fail to happen for twins.

The game theoretic formulation provides a principled way to aggregate subject-specific importances to variable importance measures for the whole data set or for subsets of special interest, using the additivity property of Shapley value.
The measure is an extension of the cohort Shapley measure
that \cite{mase:owen:seil:2019} use to quantify variable
importance for black box functions.

Other things being equal, a more
identifying variable is one that is more worth concealing
for privacy purposes, or more valuable for personalization.
That said, our measure is not designed for settings like
differential privacy \citep{dwor:2008} where one
seeks privacy guarantees.
We use it instead for exploratory purposes.

An outline of the paper is as follows. 
Section~\ref{sec:notation} introduces our notation, reviews Shapley value and defines the uniqueness Shapley values of each input variable for a given subject. Variables with greater Shapley value are more identifying.  
Section~\ref{sec:entropy} shows that the uniqueness Shapley value can be related to an entropy measure.  When aggregated to an entire data set, the uniqueness Shapley value for variable $j$ is a weighted sum of the conditional empirical  entropies of variable $j$ given all subsets of variables not including $j$. When we aggregate only over a proper subset of subjects the resulting Shapley value expression replaces entropies by cross-entropies linking the empirical distribution on the subset to the full data set.
A naive implementation of aggregated uniqueness Shapley value will have a cost that is quadratic in the number of subjects. Section~\ref{sec:adtrees} describes the all dimension trees of \cite{moor:lee:1998} that we have found give an enormous speedup making the difference between feasible and infeasible computation in some of our examples.
Section~\ref{sec:flares} explores a solar flare dataset from \cite{dua:graf:2019}. We treat solar regions with the most extreme and potentially dangerous flares as anomalies and then, as a step towards explaining those anomalies, look at which variables most identify them.
Section~\ref{sec:nc} looks at voter registration data from the state of North Carolina.  We compare the extent to which race, age, gender and other variables serve to identify voters.  
Section~\ref{sec:discussion} gives conclusions and 
discusses some further issues.

\section{Notation and background}\label{sec:notation}
We suppose that there are $d$ categorical variables
measured on each of $n$ subjects.  
Subject $i$ is described by a vector $\bsx_i$
with components $x_{ij}\in\cx_j$ for $j=1,\dots,d$.
The set of all subjects is denoted $1{:}n$ and the set of all variables is $1{:}d$.

To begin, suppose that there is a target
subject $t\in 1{:}n$, and we want to know what
variables identify subject $t$.
For every subset of variables $u\subseteq 1{:}d$
let $\bsx_{iu}$ be the tuple $(x_{ij})_{j\in u}$.
Then we define the cohort
\begin{align*}
C(u)&=C_t(u)=\{i\in 1{:}n\mid \bsx_{iu}=\bsx_{tu}\},
\end{align*}
of all subjects $i$ who match subject $t$ on every variable
in the set $u$. By convention, $C_t(\emptyset) = 1{:}n$
and no cohort is empty because they all include $t$.
The size of a cohort is the cardinality
$$
N(u) = N_t(u) = \sum_{i=1}^n1\{i\in C_t(u)\}.
$$

We consider variable $j$ to be important for identifying
subject $t$ if $N_t(u\cup\{j\})$ is typically much smaller
than $N_t(u)$ for $u\subseteq1{:}d\setminus\{j\}$.
In that case, knowledge of $x_{tj}$ refines the
cohort containing subject $t$ by a large factor.
There are $2^{d-1}$ different cohorts into which $\{j\}$
might be included, and we will use Shapley value to
incorporate them all.

To simplify some of our expressions, we
introduce some notational short forms.
The set $1{:}d\setminus \{j\}$ is 
written as $-j$. When $j\not\in u\subset 1{:}d$
then we write $u+j$ for $u\cup\{j\}$.
The set of all subsets of $1{:}d$  is
written $2^{1{:}d}$.
When $u$ is a finite set, then $|u|$ is its cardinality.
%The vectors $\bszero$ and $\bsone$ have all components
%equal to $0$ and $1$ (respectively).

\subsection{Shapley value}

The Shapley value from game theory can be used
to allocate value to the members of a team that produced the value.
In our context, the team will be made up of variables
whose values are specified, and the value will
be defined by how much subject $t$ is identified.

We work with a value function $\val\,{:}\,2^{1{:}d}\to\real$
where $\val(u)$ is the value created by the team $u$,
and we suppose that we are given $\val(u)$ for all $u\subseteq 2^{1{:}d}$.
The total value created by the team is $\val(1{:}d)$, and the
problem is to make a fair allocation to members $j=1,\dots,d$.
The fair share for member $j$ is denoted by $\phi_j$.
\cite{shap:1953} had these axioms:
\begin{compactenum}[\quad1)]
\item (Efficiency) $\sum_{j=1}^d \phi_j =\val( 1{:}d)$,
\item (Symmetry) if $\val(u+j)=\val(u+j')$ whenever $u\subseteq1{:}d\setminus\{j,j'\}$
then $\phi_j=\phi_{j'}$,
\item (Dummy) if $\val(u+j)=\val(u)$ whenever $u\subseteq-j$ then $\phi_j=0$ and
\item (Additivity) if two games have value functions $\val$ and $\val'$
and shares $\phi_j$ and $\phi'_j$ then the game with values
$\val(u)+\val'(u)$ must have shares $\phi_j+\phi'_j$.
\end{compactenum}
\smallskip

The fair shares depend strongly on the incremental value of variable $j$ given that variables $u\subseteq-j$ are already included.  It is convenient to use
\begin{align}\label{eq:incrementalvalue}
\val(j\giv u) \equiv \val(u+j)-\val(u)
\end{align}
to represent this incremental value.
\cite{shap:1953} finds that there is a unique set of shares $\phi_j$ (called Shapley values)
that satisfy his four axioms.  They are
\begin{align}\label{eq:shapval}
\phi_j = \frac1d\sum_{u\subseteq -j}{d-1 \choose |u|}^{-1}
\val(j\giv u).
\end{align}

Another way to describe $\phi_j$ is to build
a set from $\emptyset$ to $1{:}d$ by adding
the variables $j\in1{:}d$ in a random order.
At some point variable $j$ appears with a set $u\subseteq -j$
of previously introduced variables.
Then $\phi_j$ is the average of $\val(j\giv u)$
taken over all $d!$ variable orders.

We see from equation~\eqref{eq:shapval} that only value differences
affect $\phi_1,\dots,\phi_d$. It is often convenient to
take $\val(\emptyset)=0$.  If that does not hold we can
replace every $\val(u)$ by $\val(u)-\val(\emptyset)$
without changing any of the $\phi_j$.

\subsection{Uniqueness Shapley}

The value function we choose for
identifying subject $t$ is
\begin{align}\label{eq:uniquenessvalfun}
\val(u) = -\log_2\Bigl( \frac{N_{t}(u)}{N_t(\emptyset)}\Bigr).
\end{align}
This definition satisfies $\val(\emptyset)=0$.
The smaller the cardinality $N_t(u)$ of $C_t(u)$, the more $\bsx_{t,u}$
has served to identify subject $t$.
One unit of $\val(\cdot)$ corresponds to information
that halves the size of the cohort containing subject $t$.
We quantify the importance of $x_{tj}$ to the identifiability of subject $t$
via the Shapley value $\phi_{j} =\phi_{t,j}$ derived
from the value function in~\eqref{eq:uniquenessvalfun}.
The extent to which revealing $x_{tj}$ identifies subject $t$
depends on any previously identified variables $\bsx_{tu}$
for $u\subseteq-j$.  The uniqueness Shapley value
combines all $2^{d-1}$ of these contributions in a way
consistent with game theory.
Those contributions take the form
$$
\val(j\giv u) =\log_2\Bigl( \frac{N_{t}(u)}{N_t(u+j)}\Bigr)
$$
after cancellation of $\log_2(N_t(\emptyset))$.

The uniqueness Shapley value function is the cohort Shapley
value function of
\cite{mase:owen:seil:2019} after the within-cohort average of a response
variable is replaced by the cardinality of the cohort.

\begin{proposition}
The uniqueness Shapley value $\phi_{t,j}$ satisfies
$\phi_{t,j}\ge0$ with $\phi_{t,j}=0$ if and
only if $x_{ij}=x_{tj}$ for all $i=1,\dots,n$.
\end{proposition}
\begin{proof}
If $j\not\in u$, then $N_t(u+j)\le N_t(u)$ and from this
we find that $\phi_{t,j}\ge0$.
If $x_{ij}=x_{tj}$ for all $i$ then $C_t(u+j)=C_t(u)$
for $j\not \in u$ making $N_t(u)=N_t(u+j)$
and $\val(u+j)-\val(u)=0$.
Conversely, suppose that $\phi_{t,j}=0$
but $x_{ij}\ne x_{tj}$ for some $i$.
Then $N_t({\{j\}})>N_t(\emptyset)$
from which $\val(\{j\}\giv\emptyset)>0$.
This provides a contradiction because there
cannot be any compensating negative value differences
to bring the Shapley value down to zero.
\end{proof}

Now suppose that we want to quantify the importance
of variable $j$ to the whole set of subjects.  The
additivity axiom of Shapley value makes it natural
to sum those values. For interpretability we scale that
sum to an average over subjects
taking
\begin{align}\label{eq:totalvalue}
\phi^{1:n}_j = \frac1n\sum_{t=1}^n\phi_{t,j}
\end{align}
as our global measure of the cardinality importance of variable $j$.

For an arbitrary non-empty subset $v\subseteq1{:}n$ of subjects we can also define
\begin{align}\label{eq:subsetvalue}
\phi^{v}_j = \frac1{|v|}\sum_{t\in v}\phi_{t,j}.
\end{align}
Suppose for instance that $x_{ij}$ encodes the gender of subject $i$.
We can then define a set $v$ consisting of all the subjects with
one of the genders in the data set and then $\phi^v_j$
describes the extent to which gender identifies people
of the given gender.  This is not necessarily 
zero even though gender is constant over $i\in v$
because the Shapley values $\phi_{t,j}$ are defined 
on the entire
subject set $1{:}n$.  For some other feature 
$j'$ such as age in years
we can then measure the extent to which $j'$ 
identifies subjects of a given gender and see how 
this varies as
we change the set $v$ to each gender in turn.

 \section{Relationship to information theory}\label{sec:entropy}

It is natural to consider entropy
as a measure of how informative a feature is for identification. \cite{gomez2018hiding}
report entropy values for individual variables
in their browser fingerprint data. Here we introduce
some information theoretic quantities and show that
uniqueness Shapley value aggregated over subjects
is equivalent to entropy when
the features are independent.
More generally, aggregating the uniqueness Shapley 
measure yields a linear
combination of conditional entropies.  Aggregates over
proper subsets of subjects involve cross-entropies. 

For a categorical variable $\bsx\in\cx$ with $\bsx\sim p$ we write the entropy of both $\bsx$ and $p$ as $H(\bsx)=H(p)=-\sum_{\bsx\in\cx}p(\bsx) \log_2(p(\bsx))$.
For disjoint $u,\tilde u\subset1{:}d$ the conditional entropy of $\bsx_u$ given $\bsx_{\tilde u}$ is
\begin{align*}
H(\bsx_u\giv\bsx_{\tilde u})&=\sum_{\bsz_{\tilde u}\in\cx_{\tilde u}}p(\bsz_{\tilde u})
H(\bsx_u\giv\bsx_{\tilde u}=\bsz_{\tilde u}),\quad\text{where}\\
H(\bsx_u\giv \bsx_{\tilde u}=\bsz_{\tilde u})&=-\sum_{\bsx_u\in\cx_u}p(\bsx_u\giv\bsx_{\tilde u}=\bsz_{\tilde u})\log_2(p(\bsx_u\giv\bsx_{\tilde u}=\bsz_{\tilde u})).
\end{align*}
By the chain rule for entropy \citep[Theorem 2.2.1]{cove:thom:2006}
$$H(\bsx_u\giv\bsx_{\tilde u})=H(\bsx_{u\cup \tilde u})-H(\bsx_u).$$
We will work with the entropy of sub-vectors of $\bsx$ and for this we write $\ent(u)=\ent(u;\bsx)=H(\bsx_u)$ when the distribution of $\bsx_u$ is understood from context. Similarly, $\ent(u\giv \tilde u)$ denotes $H(\bsx_u\giv\bsx_{\tilde u})$. 
For $j\not\in u$ we may abbreviate 
the conditional entropy $\ent(\{j\}\giv u)$ 
to $\ent(j\giv u)$.

\subsection{Relationship to entropy}
Let $p(\cdot)$ be the empirical distribution
on $\cx$ with
\begin{align}\label{eq:empirical}
p(\bsx) = \frac1n\sum_{i=1}^n1\{\bsx=\bsx_i\}.
\end{align}
For $j\in1{:}d$, let $p_j$ be the marginal distribution of $\bsx_j$
under~\eqref{eq:empirical} and for $u\subseteq1{:}d$
let $p_u$ be the marginal distribution of $\bsx_u$ under~\eqref{eq:empirical}.

We say that two or more variables are independent in the
data if they are independent random variables under~\eqref{eq:empirical}.
Exact independence is quite unlikely to occur but it provides
an interpretable baseline via Shannon's entropy.

\begin{proposition}\label{prop:itsentropy}
Suppose that $\bsx_{\{j\}}$ is independent of $\bsx_{-j}$ under~\eqref{eq:empirical}.
Then
$$
\phi^{1:n}_j = \ent(\{j\}). 
%-\sum_{x\in\cx_j}p_j(x) \log_2(p_j(x))
$$
for $\bsx$ with the empirical distribution~\eqref{eq:empirical}.
\end{proposition}
\begin{proof}
Consider subject $t$.  Because $\bsx_{\{j\}}$ is independent of $\bsx_{-j}$
we find that $\bsz_{t,\{j\}}$ is independent of $\bsz_{t,-j}$.
This means that for all $u\subseteq-j$,
$$\frac{N_{t}(u+j)}{N_t(u)}
=p_j(x_{tj})$$
and then $\phi_{t,j} = -\log_2(p_j(x_{tj}))$ based on its
expression as an average over permutations of incremental values.
Now aggregating over subjects,
$$
\phi_j^{1:n} = \frac1n\sum_{t=1}^n-\log_2(p_j(x_{tj}))
=-\sum_{x\in\cx_j}p_j(x)\log_2(p_j(x)).\qedhere
$$
\end{proof}

As usual, the proper interpretation of $0\log_2(0)$ is zero.
If all the variables are independent, then
they all have a uniqueness Shapley value
equal to their entropy.
The connection to entropy goes further.
\begin{proposition}
The global uniqueness Shapley value for variable $j$ is
\begin{align}\label{eq:shapvalent}
\phi^{1{:}n}_j 
&= \frac1d\sum_{u\subseteq -j}{d-1 \choose |u|}^{-1}
\ent(j\giv u)
\end{align}
where the conditional entropies are computed for a random vector $\bsx$ with the empirical distribution~\eqref{eq:empirical}.
\end{proposition}
\begin{proof}
For $t\in1{:}n$ we have
$N_t(u)/N_t(\emptyset)=N_t(u)/n=p_{u}(\bsx_u)$. Then
$$
\val^{1:n}(u) = \frac1n\sum_{t=1}^n-\log_2(N_t(u)/n)=\ent(u)+\log_2(n)
$$
and so $\val^{1:n}(j\giv u)=\val^{1:n}(u+j)-\val^{1:n}(u)=\ent(u+j)-\ent(u)=\ent(j\giv u)$.
\end{proof}

Because conditional entropies are non-negative and
$\val(1{:}d)$ is the entropy of $\bsx$ under $p$, we have the bracketing inequality
\begin{align}\label{eq:bracketing}
\frac{\ent(\{j\})}d\le \phi^{1{:}n}_j\le\ent(1{:}d).
\end{align}
Noting that $\ent(\emptyset)=0$
we find for $d=2$ that
\begin{align}\label{eq:phi1of2}
\phi^{1{:}n}_1 = \frac12\ent(\{1\})+\frac12\ent(\{1\}\giv\{2\})
\end{align}
with $\phi^{1{:}n}_2$ found by switching indices.
For $d=3$,
\begin{align*}
\phi^{1{:}n}_1 &= \frac13\ent(\{1\})
+\frac16\ent(\{1\}\giv\{2\})
+\frac16\ent(\{1\}\giv\{3\})
+\frac13\ent(\{1\}\giv \{2,3\}).
\end{align*}

It may seem counterintuitive that larger entropy corresponds to greater power to identify subjects.  If a variable takes two levels, say 0 and 1, then the distribution with greatest entropy is the one that gives them each probability $0.5$.  Revealing that variable provides `1 bit' of cohort reduction. If instead there is a 90:10 split for some variable then 10\% of the population find that their cohort size is greatly reduced by 10-fold ($\log_2(10)\approx 3.3$) but 90\% find their cohort size reduced by the
much lower amount, $1/0.9$. This is about 11\%  and $\log_2(1/0.9)\approx0.15$, so the average number of bits is $0.1\times 3.3+0.9\times 0.15\approx0.46$.

The largest possible global uniqueness Shapley
value for a binary predictor variable $x_j$
is $\phi^{1:n}_j=-\log_2(1/2)=1$.  In the
solar flare example of Section~\ref{sec:flares} we will see
$\phi^v_j>1$ for a binary predictor $x_j$
and a set $v$ of anomalous observations.

\subsection{Relationship to cross entropy}

For two distributions $p(\bsx)$ and $q(\bsx)$ on a discrete set $\cx$ the relative entropy (Kullback-Leibler distance) from $p$ to $q$ is
$$
D(q\kld p)=\sum_{x\in\cx}q(\bsx)\log_2\Bigl(\frac{q(\bsx)}{p(\bsx)}\Bigr).
$$
Similarly, the cross entropy of $p$ relative to $q$ is
$$
H(p,q) = -\sum_{\bsx\in\cx}q(\bsx)\log_2(p(\bsx))
=H(q)+D(q\kld p).
$$
It is very common to write these expressions with symbols $p$ and $q$ reversed, but in our setting, the second argument needs to be the empirical distribution from~\eqref{eq:empirical} that we have labeled $p$.
For distribution $q$, we use marginal
distributions $q_j$, $q_u$ and $q_{u+j}$ analogous
to the quantities that we have used previously
for $p$.

\begin{proposition}
The uniqueness Shapley values for subset $v$ of subjects at~\eqref{eq:subsetvalue} can be written
\begin{align}\label{eq:shapcross}
\phi^v_j = \frac1d\sum_{u\subseteq -j}{ d-1\choose |u|}^{-1}\bigl( H(p_{u+j},q_{u+j})-H(p_u,q_u)\bigr)
\end{align}
where $q$ is the uniform distribution on $\bsx_t$ for $t\in v$.
\end{proposition}
\begin{proof}
First, by definition
\begin{align*}
\phi^v_j &=\frac{1}{|v|}\sum_{t\in v}\frac1d\sum_{u\subseteq -j}{ d-1\choose |u|}^{-1}\val(j\giv u) \\
&=\frac1d\sum_{u\subseteq -j}{ d-1\choose |u|}^{-1}\frac{1}{|v|}\sum_{t\in v}\log_2\left(\frac{N_t(u)}{n_t(u+j)}\right).
\end{align*}

Next, for a set $w\subseteq 1{:}d$,
$$H(p_w,q_w) = -\sum_{\bsx\in\cx}q_w(\bsx)\log_2(p_w(\bsx)).$$
This sum can be rewritten over $t\in v$ as long as we divide out the multiplicity for each $\bsx$ among the $\bsx_t$ for $t\in v$ yielding
$$H(p_w,q_w) = -\sum_{t\in v}\frac{1}{N^v_t(w)}q_w(\bsx_t)\log_2(p_w(\bsx_t)).$$
Here $N^v_t(w)=\sum_{i \in v}1\{i\in C_t(w)\}$ is the number of subjects in $v$ that match $\bsx_t$ on the features in $w$.  As $q_w$ is the uniform distribution over $\bsx_t$ for $t\in v$, $q_w(\bsx_t)=(N^v_t(w)/|v|)$, i.e., the proportion of subjects in $v$ that match $\bsx_t$ on the features in $w$. Therefore
$$H(p_w,q_w) = -\sum_{t\in v}\frac{1}{|v|}\log_2\left(\frac{N_t(w)}{n}\right).$$

With this formulation we can see that
\begin{align*}
    & \frac1d\sum_{u\subseteq -j}{ d-1\choose |u|}^{-1}\left(H(p_{u+j},q_{u+j})-H(p_u,q_u)\right)\\
   = & \frac1d\sum_{u\subseteq -j}{ d-1\choose |u|}^{-1}\left(\left(-\sum_{t\in v}\frac{1}{|v|}\log_2\left(\frac{N_t(u+j)}{n}\right)\right)-\left(-\sum_{t\in v}\frac{1}{|v|}\log_2\left(\frac{N_t(u)}{n}\right)\right)\right)\\
   = &\frac1d\sum_{u\subseteq -j}{ d-1\choose |u|}^{-1}\frac{1}{|v|}\sum_{t\in v}\log_2\left(\frac{N_t(u)}{n_t(u+j)}\right) = \phi^v_j.\qedhere
\end{align*}
\end{proof}

\subsection{Duplicate and redundant variables}
If some other variable $j'$ is
equivalent to variable $j$, then it must have the
same uniqueness Shapley value.  The extreme version of this is that
the second variable could have been copy-pasted from
the first one by accident.  More plausibly there could
be two variables like a person's email address and cell phone
number that have very nearly a one to one relationship
in a data set of transactions.

Let's look at $\phi_1^{1:3}$ in the special
case where all $x_{i3}=x_{i1}$.
Then $\ent(\{1,3\})=\ent(\{1\})=\ent(\{3\})$
and $\ent(\{1,2\})=\ent(\{2,3\})=\ent(\{1,2,3\})$
and we find that
\begin{align*}
\phi^{1{:}s}_1 &= \frac13\ent(\{1\})
+\frac16\ent(\{1\}\giv\{2\}).
\end{align*}
This is strictly smaller than
the value in~\eqref{eq:phi1of2} unless $\bsx_{\{1\}}$ is constant in the data
in which case both are zero.

Next we consider the effect of introducing this duplicate on $\phi_2^{1:n}$.
We get
\begin{align*}
\phi^{1{:}n}_2 &= \frac13\ent(\{2\})
+\frac16\ent(\{2\}\giv\{1\})
+\frac16\ent(\{2\}\giv \{3\})
+\frac13\ent(\{2\}\giv \{1,3\})\\
& = \frac13\ent(\{2\})
+\frac23\ent(\{2\}\giv\{1\}).
\end{align*}
The coefficient of $\ent(\{2\})$ has decreased
from $1/2$ to $1/3$ while the coefficient
of $\ent(\{2\}\giv\{1\})$ has increased
from $1/2$ to $2/3$.
% I think this means the duplicate of variable 1 could either increase OR decrease the importance of variable 2

A redundant variable is one that can be perfectly identified based on some subset of other variables.  Like a duplicated variable, the redundant ones do not get a uniqueness Shapley value of zero.

\subsection{Database key}

Suppose that one of the variables uniquely identifies each
subject.  That is, we never have $x_{ij}=x_{tj}$
unless $i=t$.  We can think of this as the key variable
in a database or even the row number in a data frame.
The presence of this variable 
forces $\val(1{:}d)=\log_2(n)$ (it's largest possible value)
for every subject.
For each subject,  $N_t(u)=1$ whenever the key is in $u$.

We can look at $\val(j\giv u)$
where $u\subseteq 1{:}d$ is a set of
variables and now suppose that we introduce
a database key with index $j=0$.
For each of the prior $d!$ orders in which
variable $j\ge1$ could have been included
there are $d+1$ positions at which the new
variable $j=0$ could be introduced.
If variable $0$ is introduced after variable $j$
then the incremental value is unchanged.
If variable $0$ is introduced before variable $j$
then the new incremental value for variable $j$ is $0$.
Therefore introducing the key
changes the uniqueness Shapley value for
one subject to
\begin{align}\label{eq:shapvaliwithkey}
\phi_j &= \frac1d\sum_{u\subseteq -j}{d-1 \choose |u|}^{-1}
\Bigl(\frac{d-|u|}{d+1}\Bigr)\val(j\giv u)\notag\\
&=\frac1{d+1}\sum_{u\subseteq -j}{d\choose |u|}^{-1}
\val(j\giv u)
\end{align}
because $d-|u|$ of the $d+1$ possible insertion points
preserve the incremental value while the others remove it.
Equation~\eqref{eq:shapvaliwithkey} holds for $v\subseteq1{:}n$
with $\phi_{t,j}$ corresponding to $v=\{t\}$.
The sum in~\eqref{eq:shapvaliwithkey} is taken over subsets $u$ of the original $d$ variables exclusive of both variable $j$ and the posited key variable $0$.

After introducing the key variable, the contribution of $\val(j\giv u)$ is downweighted by a factor of $(d-|u|)/(d+1)$, which ranges from $d/(d+1)$ for $u=\emptyset$ to $1/(d+1)$ for $u=-j$. The average value of this factor over sets $u\subseteq-j$ is
$$
\frac1d\sum_{r=0}^{d-1}\frac{d-r}{d+1}=\frac12.
$$
Other things being equal we might expect that introducing the database key will halve the other variables' uniqueness Shapley values.  Variables that get their importance mostly from $\val(j\giv u)$ for small $|u|$ are less affected by the key, and variables that get their importance mostly from $\val(j\giv u)$ for large $|u|$ will lose more than half of their uniqueness Shapley values.

% \subsection{Marichal entropy}
% \cite{mari:roub:2000} present two notions of entropy based on Shapley value.
% Both of them are at the aggregate level and not specific to individual subjects.
% They pertain to fuzzy measures $\mu$ on the set $1{:}d$ of indices for which  $\mu(\emptyset)=0$, $\mu(1{:}d)=1$ and $\mu(u')\ge\mu(u)$ whenever $u\subseteq u'$.
% Given a set of non-negative Shapley values $\phi_j$ summing to one, derived from the fuzzy measure, the upper Marichal entropy is
% $$
% -\sum_{j=1}^d \phi_j\log_2(\phi_j).
% $$
% That is, upper Marichal entropy is the entropy of the 
% Shapley values $(\phi_1,\dots,\phi_d)$ when those 
% values form a probability distribution.
% It is thus an entropy of Shapley values, not a Shapley value defined by incremental entropies as we have here.
% The upper Marichal entropy is
% $$
% \frac1d\sum_{j=1}^d\sum_{u\subseteq -j} {d-1\choose |u|}^{-1}h( \mu(u+j)-\mu(u))
% $$
% where $h(x)=-x\log_2(x)$.  It is a Shapley value 
% where the incremental contribution from $j$ is the 
% entropy of a difference of probabilities, 
% not a difference of entropies.

\section{All dimension trees}\label{sec:adtrees}

It would require $O(n|u|)$ time to compute $N_t(u)$ by naively checking which subjects match the target $t$ on a particular subset of features $u$.  For a naive calculation of the uniqueness Shapley, we would therefore need to compute $N_t(u)$ $O(2^d)$ times for each combination of feature and subjects to calculate each $\phi_{j}^{t}$.  Therefore, the full run time of the naive implementation of uniqueness Shapley would be $O(n^2 d^22^d)$.  For a reasonably large number of subjects, the $n^2$ factor can be computationally prohibitive even when $d$ is not large.  To improve upon the naive implementation, we employ a more suitable data structure: the all dimension tree of \cite{moor:lee:1998}.

The all dimension tree is optimized for tasks similar to calculating $N_t(u)$, i.e., generating contingency tables.  The basic structure takes categorical data and constructs a tree where each branch corresponds to a particular ``feature equals value" query, and the subsequent node stores the count of all subjects for whom that query and all preceding queries in the tree are true.  To compute $N_t(u)$ for a given $t,u$ pair, you start at the root of the tree and follow the branches corresponding to queries $x_j=x_{tj}$ for all $j \in u$, and the count stored in the resulting node would be $N_t(u)$.  Therefore, with such a tree structure, we only require $O(|u|)$ time to compute $N_t(u)$ which no longer scales with the number of subjects.  Note that the same tree can be used for all subjects, so it need not be constructed more than once for a single data set.

A naive version of this structure would generally require a prohibitive amount of memory, thereby rendering the computational savings moot even for relatively small $d$.  The innovation of the all dimension tree comes from its techniques to reduce memory requirements especially in the common cases of sparsity and correlated features.  They use several different methods to achieve this goal which are not relevant to the scope of this discussion except that they succeed in reducing the memory cost.  For example, for binary features, the memory requirement is $O(2^d)$ in the worst case compared to $O(3^d)$ in the dense naive implementation, and it can achieve $O(d)$ in the best case.  Performance closer to the best case is achieved when features are distributed more unevenly and are more correlated, resulting in a sparser distribution.  The time to initially construct the all dimension tree is also linear in $n$ and while worst-case exponential in $d$, it is not worse than our dependence to $d$ in the Shapley calculation.  This makes our overall implementation of uniqueness Shapley using all dimension trees $O(n d^22^d)$.  A speed up by a factor of $n$.  We use the all dimension tree Python implementation developed by \cite{ding:2018}.

\begin{algorithm}[t]
	\caption{Uniqueness Shapley Pseudo Code} 
	\begin{algorithmic}[1]
	    \State input: feature matrix X
	    \State T=ADTree(X)
	    \State Shap=zeros(n,d)
		\For {subject $i=1,2,\ldots,n$}
			\For {feature $j=1,2,\ldots,d$}
			    \For {$u\subseteq -j$}
			        \State $\gamma=\frac{1}{d{d-1\choose|u|}}$
			        \State $N_1$=T.query(i,u)
			        \State $N_2$=T.query(i,u+j)
			        \State Shap[i,j]$+=\gamma\log(N_1/N_2)$
		    	\EndFor
		    \EndFor
		\EndFor
		\State Return Shap
	\end{algorithmic} 
\end{algorithm}

\begin{table}[t]
\centering
\begin{tabular}{lccrr} % Journals usually don't want vertical lines in tables
\toprule
   {Data}&\mca{n} &\mca{d} &\mca{ADTree}&\mca{Standard}\\\toprule
  {Solar Flare}& \phantom{7,53}1,066&9& 5.30&54.55
   \\
   \midrule
        {Dare County Census}&\phantom{7,0}30,921&5&  1.91 & 2,522.40
   \\
   \midrule
        {Durham County Census}&\phantom{7,}253,563&5&16.46 & 32,426.91
   \\
   \midrule
     {North Carolina Census}&7,538,125&5&362.51 & n/a
   \\
   \bottomrule
\end{tabular}
\caption{Run times in seconds for uniqueness Shapley for the solar flare data, the full North Carolina census data, and two specific counties of North Carolina census data.}
\label{tab:shaprun}
\end{table}

For a reasonably large number of features $d$, this implementation is no longer computationally feasible.  In those cases, a Monte Carlo approximation must be employed as in \cite{male:2013}.  The all dimension tree structure can also be adapted for large $d$ to reduce its memory dependence at the cost of only approximately calculating $N_t(u)$ which is also discussed in \cite{moor:lee:1998}.  Our present examples did not have such very large $d$.  An implementation of uniqueness Shapley can be found on our GitHub: \url{https://github.com/cohortshapley/uniquenessshapley}.

\section{Solar flare data example}\label{sec:flares}

Many data sets have a few entries that are anomalies, such as outliers.  There have been many efforts to detect anomalies and others to explain them.  
For a survey of anomaly detection, see \cite{chan:bane:kuma:2009}.
\cite{jaco:etal:2020} provide the Exathlon benchmark for anomaly explanation methods, aimed at time series. They include a method based on marginal entropies.

In this section we look at the solar flare data set from the UC Irvine repository \citep{dua:graf:2019}. Some of the solar flares
have been marked as unusual (anomaly detection).
We consider uniqueness Shapley as a way to 
understand which variables make the anomalous
data most unique. We must add that finding
an identifying variable for an anomaly is a
kind of association and is not necessarily causal.
For instance, an unusual person's social
security number is very identifying but 
is unlikely to be causal.

We use the second solar flare data set from the UC Irvine repository because it is said to be more reliable. It describes $n=1066$ regions on the surface of the sun. There are 10 categorical predictors and 3 responses indicating the number of common (C class), moderate  (M class) and severe (X class) solar flares in each region over a 24 hour period. Severe flares are 100 times as strong as moderate ones which in turn are 10 times as strong as common ones. The M class flares can cause radio blackouts or endanger astronauts. There are also numerical gradations within these classes. See \url{https://www.nasa.gov/mission_pages/sunearth/news/X-class-flares.html}.

Some of those solar regions are much more interesting than others. All but five of them had no severe flares. Of those five, one had two severe flares.   Four of the regions had three or more flares rated moderate or severe, so we consider those too. 
We will use uniqueness Shapley to study these anomalies in terms of the categorical predictors of the data set.
The first nine categorical predictor variables are described in Table~\ref{tab:flarevars}.  A tenth variable, about the area of the largest spot, was constant for all 1066 regions. We omit that variable
and work with $d=9$ others.

\begin{table}
\centering
\begin{tabular}{ll}
\toprule
Variable & Levels\\
\midrule
Modified Zurich class & (A, B, C, D, E, F, H)\\
Largest spot size & (X, R, S, A, H, K)\\
Spot distribution & (X, O, I, C)\\
Activity &(1 = reduced, 2 = unchanged)\\
Evolution & (1 = decay, 2 = no growth, 3 = growth)\\
Prior 24 hr activity  &(1 = no $\ge$ M1s, 2 = one $\ge$ M1, 3 = multiple $\ge$ M1s)\\
Historically-complex &(1 = yes, 2 = no)\\
This pass& (1 = yes, 2 = no)\\
Area & (1 = small, 2 = large)\\
\bottomrule
\end{tabular}
\caption{\label{tab:flarevars}
Nine solar flare predictor variables. The `this pass' variable answers the question: Did the region become historically complex on this pass across the sun's disk? Source:
\protect\url{https://archive.ics.uci.edu/ml/datasets/Solar+Flare}}
\end{table}
% 1. Code for class (modified Zurich class) (A,B,C,D,E,F,H)
% 2. Code for largest spot size (X,R,S,A,H,K)
% 3. Code for spot distribution (X,O,I,C)
% 4. Activity (1 = reduced, 2 = unchanged)
% 5. Evolution (1 = decay, 2 = no growth, 3 = growth)
% 6. Previous 24 hour flare activity code (1 = nothing as big as an M1, 2 = one M1, 3 = more activity than one M1)
% 7. Historically-complex (1 = Yes, 2 = No)
% 8. Did region become historically complex on this pass across the sun's disk (1 = yes, 2 = no)
% 9. Area (1 = small, 2 = large)
% 10. Area of the largest spot (1 = <=5, 2 = >5) 
The first two columns of Table~\ref{tab:solarshapley} show
entropy and cardinality
Shapley values for the 9 predictors
in the solar flare data.  
Zurich and large spot
are the most identifying.  Area is least 
identifying.  In this data, many of the 
uniqueness Shapley values are below their corresponding marginal entropy.

\begin{table}[t]
    \centering
    \begin{tabular}{lccccc}
\toprule
Variable & Entropy &Shapley&Common&Moderate&Severe  \\
\midrule        
Zurich & 1.64 & 1.37 & 1.72 & 1.62 & 1.81\\
Large Spot & 1.52 & 1.55&1.69 & 1.62 & 1.23\\
Spot Dist & 1.16 & 0.92& 1.17 & 1.34 & 1.49\\
Activity & 0.43 & 0.45& 0.73 & 0.67 & 0.86\\
Evolution & 0.90 & 1.17 & 0.98 & 0.98 & 0.85\\
Prev Activ& 0.18 & 0.16 & 0.33 & 0.43 & 0.93\\
Complex & 0.67 & 0.77 & 0.69 & 0.69 & 0.29\\
This Pass & 0.38 & 0.35 & 0.15 & 0.11 & 0.03\\
Area & 0.12 & 0.07& 0.22 & 0.43 & 1.36\\
\bottomrule
\end{tabular}
    \caption{
The first two columns give entropy
and uniqueness Shapley for the nine
solar flare predictors.  The
next three columns cover three
types of anomalies in increasing
order of severity as described
in the text. The `Area' variable 
takes on increasing importance
as severity increases while
`This Pass' decreases.
    }
    \label{tab:solarshapley}
\end{table}

%Solar Flare values with row number (Index) inlcuded
%Index6.17
%Zurich0.83
%Large_Spot0.87
%Spot_Dist0.57
%Activity0.25
%Evolution0.61
%Prev_Activity0.09
%Complex0.41
%This_Pass0.2
%Area0.05

%  \begin{table}[t]
%  \centering
%  \begin{tabular}{c|c|c|c|c|c|c|c|c|c|}
%   \mca{}&\mca{Zurich} &\mca{Large Spot} &\mca{Spot Dist}&\mca{Activity}&\mca{Evolution}\\\cline{2-6}
%   \mcb{Shapley}& 1.37&1.55&0.92&0.45&1.17\\\cline{2-6}
%      \mcb{Entropy}&1.64&1.52&1.16&0.43&0.90\\\cline{2-6}
%   \mca{}&\mca{Prev Activity}&\mca{Complex}&\mca{This Pass}&\mca{Area}\\\cline{2-5}
%   \mcb{Shapley}&0.16&0.77&0.35&0.07\\\cline{2-5}
%      \mcb{Entropy}&0.18&0.67&0.38&0.12\\\cline{2-5}
%  \end{tabular}
%  \caption{Solar flare average uniqueness Shapley and marginal entropy.}
%  \label{tab:sfsum}
%  \end{table}

The last three columns of 
Table~\ref{tab:solarshapley} shows cardinality
Shapley for some subsets of solar regions
of increasingly anomalous nature.
They are those with at least one 
common flare,
at least one moderate flare, and, finally, 
at least one severe flare.  
The area variable which is so unimportant
globally becomes ever more identifying
for these anomalies.  
So does `previous activity'.
These variables are associated with anomalous
solar behavior in that the more extreme
the behavior, the more identifying these become.

%   \begin{table}[t]
% \centering
% \begin{tabular}{c|c|c|c|c|c|c|c|c|c|}
%   \mca{}&\mca{Zurich} &\mca{Large Spot} &\mca{Spot Dist}&\mca{Activity}&\mca{Evolution}\\\cline{2-6}
%     \mcb{Common}& 1.72&1.69&1.17&0.73&0.98\\\cline{2-6}
%      \mcb{Moderate}&1.62&1.62&1.34&0.67&0.98\\\cline{2-6}
%   \mcb{Severe}&1.81&1.23&1.49&0.86&0.85\\\cline{2-6}
%   \mca{}&\mca{Prev Activity}&\mca{Complex}&\mca{This Pass}&\mca{Area}\\\cline{2-5}
%  \mcb{Common}&0.33&0.69&0.15&0.22
%   \\\cline{2-5}
%   \mcb{Moderate}&0.43&0.69&0.11&0.43
%   \\\cline{2-5}
%       \mcb{Severe} &0.93&0.29&0.03&1.36
%   \\\cline{2-5}
% \end{tabular}
% \caption{Solar flare average uniqueness Shapley for subgroups with more than zero common, moderate, or severe flares respectively.}
% \label{tab:sfsub}
% \end{table}

\section{North Carolina voter registration data}\label{sec:nc}

Here we consider a real world demographic example in the vein of Sweeney's original work \citep{swee:2000}.  The state of North Carolina publishes voting and demographic information about all of its registered voters each election.  Some information is withheld to maintain privacy such as exact birth dates, exact addresses, and social security numbers.  From the available features, we looked at zip code, age, race, gender, and political party affiliation for $n\approx 7.5\times10^6$ registered voters in the state.  Summary uniqueness Shapley values and baseline marginal entropy values can be found in Table~\ref{tab:ncsum}.  There we can see that the relative ordering of the uniqueness Shapley values is consistent with the entropy values, but that we have large positive deviations for zip code and age. All variables had Shapley value greater than their entropy. Recall that the lower bound on the Shapley value is one fourth of the entropy, since $d=4$.

    \begin{table}[!h]
\centering
\begin{tabular}{cccccc}
\toprule
   \mca{}&\mca{Zip Code} &\mca{Race} &\mca{Party}&\mca{Gender}&\mca{Age}\\
\midrule   
   %\cline{2-6}
  \mca{Shapley}& 8.58&1.22&1.48&1.17&5.39
   \\%\cline{2-6}
    \mca{Entropy}&  6.08&1.08&1.14&0.92&4.24
   \\
\bottomrule%   \cline{2-6}
\end{tabular}
\caption{North Carolina voter registration average uniqueness Shapley and marginal entropy.  Source: \protect\url{https://www.ncsbe.gov/results-data/voter-registration-data}}
\label{tab:ncsum}
\end{table}

In Table~\ref{tab:ncdive}, we can see the average uniqueness Shapley values for various subpopulations of the state.  Some patterns follow logically from the relative class sizes, as for example, members of races that make up a smaller percentage of the population have larger uniqueness Shapley values for race on average.  Clearly membership in a less common class should help to uniquely identify someone.  Other noticeable patterns highlight the effects of the correlation structure not captured by the marginal entropy measures.  For example, the average cardinality value for race decreases for older age cohorts.

\begin{table}[t]
\centering
\begin{tabular}{lcccccc}
\toprule
   {Subgroup}&{Zip Code} &{Race} &{Party}&{Gender}&{Age}&{\% Pop.}\\
   \midrule
  {Democratic}& 8.44&1.54&1.31&1.14&5.47&36
   \\
    {Libertarian}&  8.15&1.16&6.65&1.14&4.54&\phz1
   \\
    {Republican}&  8.76&0.72&1.51&1.15&5.50&30
   \\
       {None}&  8.57&1.32&1.51&1.22&5.23&33
   \\
   \midrule
     {Female}& 8.58&1.18&1.46&0.97&5.46&49
   \\
    {Male}&  8.61&1.15&1.49&1.20&5.40&42
   \\
    {No Answer}&  8.42&1.76&1.53&2.24&4.91&\phz8
   \\
   \midrule
     {White}& 8.75&0.53&1.58&1.09&5.52&63
   \\
    {Black}&  8.28&1.70&1.14&1.10&5.35&21
   \\
    {Asian}&  7.56&5.44&1.44&1.09&5.01&\phz1
   \\
    {Native American}&  7.64&5.14&1.40&1.01&5.13&\phz1
   \\
       {No Answer}&  8.45&2.27&1.54&1.91&4.94&10
   \\
   \midrule
       {Age $<32$}&  8.49&1.52&1.52&1.26&4.63&25
   \\
    {32--48 }&  8.48&1.33&1.50&1.20&5.22&25
   \\
    {48--63}&  8.61&1.13&1.46&1.15&5.33&25
   \\
       {$>$ 63}&  8.71&0.90&1.44&1.09&6.34&25
   \\
\bottomrule
\end{tabular}
\caption{North Carolina voter registration data's
average uniqueness Shapley values for variables
aggregated over voter subgroups.  
Final column is percentage of total population.}
\label{tab:ncdive}
\end{table}

% \begin{table}[t]
% \centering
% \begin{tabular}{c|c|c|c|c|c||c|}
%   \mca{}&\mca{Zip Code} &\mca{Race} &\mca{Party}&\mca{Gender}&\mca{Age}&\mca{\% Pop.}\\\cline{2-7}
%   \mcb{Dem}& 8.44&1.54&1.31&1.14&5.47&36
%   \\\cline{2-7}
%     \mcb{Lib}&  8.15&1.16&6.65&1.14&4.54&1
%   \\\cline{2-7}
%     \mcb{Rep}&  8.76&0.72&1.51&1.15&5.50&30
%   \\\cline{2-7}
%       \mcb{None}&  8.57&1.32&1.51&1.22&5.23&33
%   \\\cline{2-7}
%   \\\cline{2-7}
%      \mcb{Female}& 8.58&1.18&1.46&0.97&5.46&49
%   \\\cline{2-7}
%     \mcb{Male}&  8.61&1.15&1.49&1.20&5.40&42
%   \\\cline{2-7}
%     \mcb{No Ans.}&  8.42&1.76&1.53&2.24&4.91&8
%   \\\cline{2-7}
%   \\\cline{2-7}
%      \mcb{White}& 8.75&0.53&1.58&1.09&5.52&63
%   \\\cline{2-7}
%     \mcb{Black}&  8.28&1.70&1.14&1.10&5.35&21
%   \\\cline{2-7}
%     \mcb{Asian}&  7.56&5.44&1.44&1.09&5.01&1
%   \\\cline{2-7}
%     \mcb{Native Am.}&  7.64&5.14&1.40&1.01&5.13&1
%   \\\cline{2-7}
%       \mcb{No Ans.}&  8.45&2.27&1.54&1.91&4.94&10
%   \\\cline{2-7}
%   \\\cline{2-7}
%       \mcb{Age$<32$}&  8.49&1.52&1.52&1.26&4.63&25
%   \\\cline{2-7}
%     \mcb{32-48 }&  8.48&1.33&1.50&1.20&5.22&25
%   \\\cline{2-7}
%     \mcb{48-63}&  8.61&1.13&1.46&1.15&5.33&25
%   \\\cline{2-7}
   
%       \mcb{63+}&  8.71&0.90&1.44&1.09&6.34&25
%   \\\cline{2-7}

% \end{tabular}
% \caption{North Carolina voter registration average uniqueness Shapley for subgroups.  Final column is percentage of total population.}
% \label{tab:ncdive}
% \end{table}

To further investigate the relationship between class imbalance and the uniqueness Shapley value, we look at zip code defined subpopulations in Figure~\ref{fig:ncrace2}.  We plot the Shapley value for race versus the proportion of the population that does not identify as white.
There are a handful of large positive outliers to this trend.  Upon further inspection, we can note that these are all zip codes with a very high proportion of American Indian voters.

%     \begin{figure}[h]
%     \begin{center}
%     \includegraphics[width=\textwidth,height=0.35\textheight,keepaspectratio]{zipcoderace.jpg}
%     \caption{uniqueness Shapley for race versus percent giving a race other than white. One point per zip code in North Carolina. A reference
%     curve is proportional to square root.
% }
%     \end{center}
%     \label{fig:ncrace}
%     \end{figure}

       \begin{figure}[h]
    \begin{center}
    \includegraphics[width=\textwidth,height=0.35\textheight,keepaspectratio]{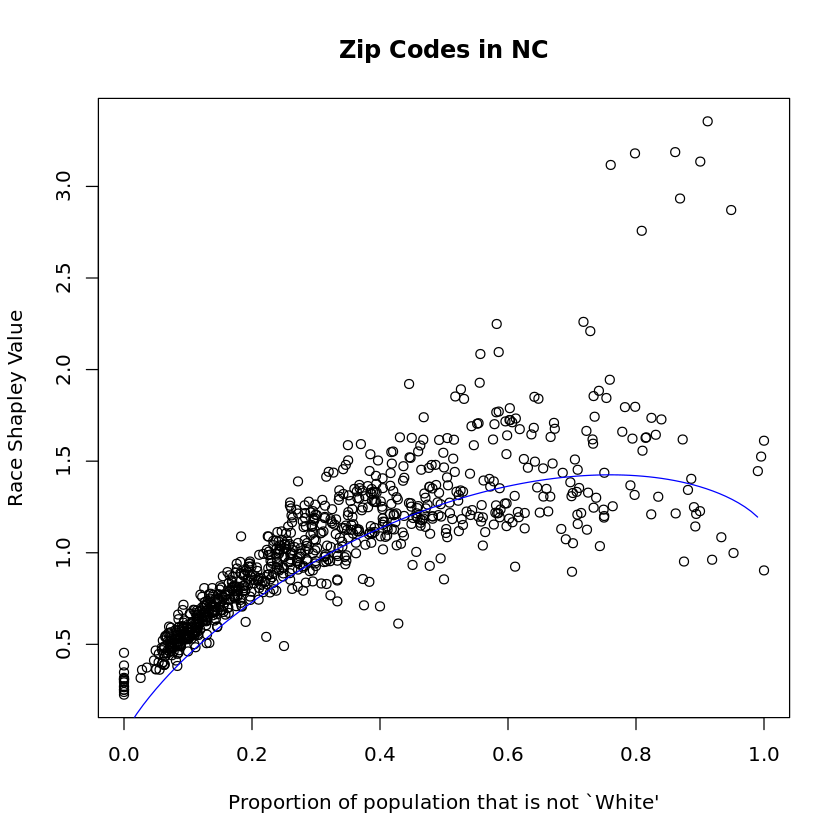}
    \caption{ \label{fig:ncrace2}
    Average uniqueness Shapley for race in a zip code versus proportion of the population in that zip code that is not `White'. The reference shows entropy of race for a hypothetical population sampled independently from North Carolina except that it has
    the implied proportion `White'.}
    \end{center}
    \end{figure}

    %   \begin{figure}[h]
    % \begin{center}
    % \includegraphics[width=\textwidth,height=0.35\textheight,keepaspectratio]{zipwent.jpeg}
    % \caption{Average uniqueness Shapley for race in a zip code versus proportion of the population in that zip code that is not `White'. A reference
    % curve is ...}
    % \end{center}
    % \label{fig:ncrace2}
    % \end{figure}

We can also measure the effect of feature granularity on the uniqueness Shapley values.  
Instead of recording location by zip code, we can coarsen it to the county level or even the whole state.
The results for this adjusted dataset are in Table~\ref{tab:ncgran}.
The uniqueness Shapley values for the other variables are not sensitive to this coarsening.
Similarly, instead of considering each age by year, we can coarsen the grouping level to five or ten or $25$ years and again we see no meaningful changes outside of the expected reduction in the age feature, even if we lump all ages
into one bucket.

\begin{table}[t!]
\centering
\begin{tabular}{lccccc}
\toprule
  {}& {Location} &{Race} &{Party}&{Gender}&{Age}\\
  \midrule
  {Zip Code}&8.58&1.22&1.48&1.17&5.39\\
   {County}&5.55&1.27&1.50&1.19&5.43\\
   {State} & 0.00 & 1.32 & 1.53 & 1.19 & 5.45\\
   \midrule
     {}& {Zip code} &{Race} &{Party}&{Gender}&{Age}\\
       \midrule
    {Single year}&8.58&1.22&1.48&1.17&5.39\\
         {5 year age buckets}&8.61&1.24&1.50&1.18&3.34
   \\
   {10 year}&8.62&1.24&1.50&1.19&2.39
   \\
      {25 year}&8.63&1.25&1.51&1.19&1.19
   \\
One bucket & 8.64  & 1.25 & 1.51 & 1.19& 0.00\\
   \bottomrule
\end{tabular}
 \caption{North Carolina voter registration average uniqueness Shapley. The first block coarsens location from zip code to county to state.  The second block coarsens age, starting
 with the original single year granularity.
 }
\label{tab:ncgran}
\end{table}

% \begin{table}[t]
% \centering
% \begin{tabular}{c|c|c|c|c|c|}
%   \mca{}& \mca{Zip Code} &\mca{Race} &\mca{Party}&\mca{Gender}&\mca{Age}\\\cline{2-6}
%   \mcb{Full Granularity}&8.58&1.22&1.48&1.17&5.39
% \\\cline{2-6}
%   \mcb{County Level}&5.55&1.27&1.50&1.19&5.43\\\cline{2-6}
%      \mcb{5y Age Buckets}&8.61&1.24&1.50&1.18&3.34
%   \\\cline{2-6}
%   \mcb{10y}&8.62&1.24&1.50&1.19&2.39
%   \\\cline{2-6}
%       \mcb{25y}&8.63&1.25&1.51&1.19&1.19
%   \\\cline{2-6}
% \end{tabular}
%  \caption{North Carolina voter registration average uniqueness Shapley for varying levels of granularity of age as well as location (zip code vs county).}
% \label{tab:ncgran}
% \end{table}

Beyond summary tables, we can visualize individual values in plots such as Figure~\ref{fig:ncshaps}.  Each vertical line corresponds to a single voter with the height of each colored segment representing the uniqueness Shapley value for the respective feature.  They are ordered by their overall uniqueness, and we only display every hundredth voter so they fit in the figure.

    %     \begin{figure}[h]
    % \begin{center}
    % \includegraphics[width=1.5\textwidth,height=0.35\textheight,keepaspectratio]{stateshaps.jpg}
    % \caption{\label{fig:ncshaps-old} % we need the label inside the caption to get the right figure number
    % Individual Shapley values stacked vertically (subsampled every 100).}
    % \end{center}
  
    % \end{figure}

           \begin{figure}[h]
     \begin{center}
     \includegraphics[height=0.3\textheight,keepaspectratio]{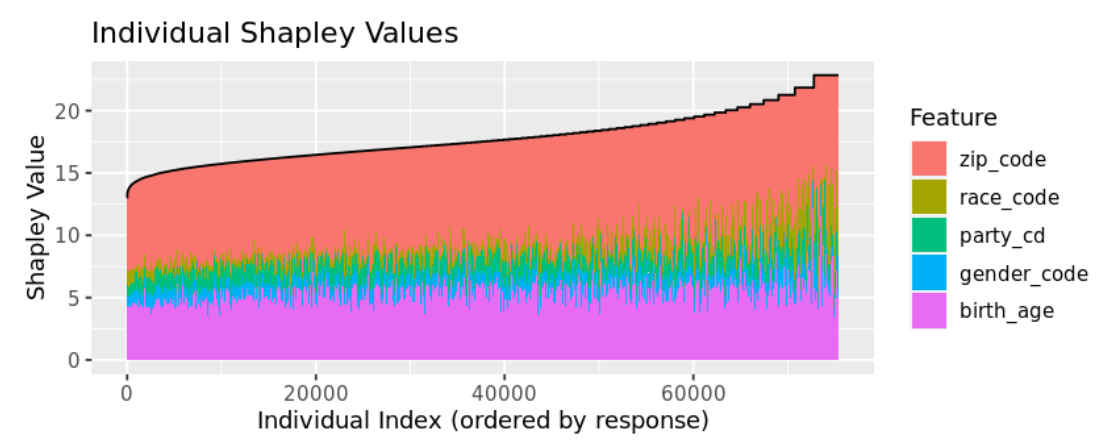}
     \caption{\label{fig:ncshaps} % we need the label inside the caption to get the right figure number
     Individual Shapley values stacked vertically (subsampled every 100).}
     \end{center}
  
     \end{figure}

\section{Discussion}\label{sec:discussion}

We have used Shapley value to quantify and compare the power that different categorical variables have to identify a subject.  When we use the additive property of Shapley value to average this measure over a population we get an expression that equals a weighted sum of conditional entropies.  Variables that when revealed increase entropy are the ones that most identify subjects.  When we average over a different distribution, such as a sub-population of interest, the entropies are replaced by cross entropies.

The extent to which a variable makes you unique
depends on the order in which it and other
variables are revealed.  Some variables might, if revealed
last, be very identifying.  Other variables might be 
redundant if revealed last but very informative if
revealed early due to associations among the variables.
The Shapley formulation combines 
all of the orders
in which a variable might be revealed.

In this work we have kept to data sets with a modest number $d$ of variables because computation of Shapley value can include a cost that grows proportionally to $d2^d$, the number of cohorts a subject might belong to.  There are Monte Carlo sampling algorithms for Shapley value that allow larger $d$.

We have focused on categorical variables.  Continuous 
variables can be coarsened into categorical ones by
setting ranges. The finer the range the greater the Shapley value is. This is appropriate because finer classifications really are more revealing.

It is also possible
to use asymmetric notions of symmetry
for continuous variables as considered in \cite{mase:owen:seil:2019}.
For instance if we declare $x_{ij}$ to
be similar to $x_{tj}$ whenever $|x_{ij}-x_{tj}|\le \delta_j|x_{tj}|$
we might find that $x_{ij}$ is similar to $x_{tj}$
but not the converse.  The proper way to account for a continuous variable when quantifying uniqueness depends on how we expect that variable might be revealed.  We can compare the effects of revealing age in 1 or 5 or 10 year windows and can also measure how the effect of revealing another variable such as race or gender depends on the granularity with which age has been revealed.

% \section*{odds and ends}
% Thanks Ben.  This is all we need.  We don't have to do a disaggregated Marichal entropy, at least not in this one.  We just needed to show that uniqueness Shapley is not the same.
% Marichal notion of entropy for fuzzy measures.  Suppose $\mu$ is a fuzzy measure on $N$, let $h(x)=-x\log x$, and let $\gamma_u=\frac{1}{n}{n-1 \choose |u|}^{-1}$.

% Lower Marichal entropy:

% $$H_{l}(\mu)=\sum_{j\in N}\sum_{u\subseteq -j}\gamma_{u}h(\mu(u+j)-\mu(u))$$

% Upper Marichal entropy:

% $$H_{u}(\mu)=\sum_{j\in N}h\left(\sum_{u\subseteq -j}\gamma_{u}(\mu(u+j)-\mu(u))\right)=\sum_{j\in N}h(\phi_j(\mu))$$

% We could define a disaggregate lower Marichal entropy:

% $$H^{j}_{l}(\mu)=\sum_{u\subseteq -j}\gamma_{u}h(\mu(u+j)-\mu(u))$$

% In this notation we have the uniqueness Shapley in the form:

% $$\phi^{1:s}_{j}=\sum_{u\subseteq -j}\gamma_{u}\mathcal{H}(j|u)=\sum_{u\subseteq -j}\gamma_{u}\left(\mathcal{H}(u+j)-\mathcal{H}(u)\right)$$

\section*{Acknowledgments}

This work was supported
by the U.S.\ National Science Foundation
under projects IIS-1837931 and DMS-2152780
and by a grant from Hitachi, Ltd.
We thank two anonymous reviewers for helpful comments.

\bibliographystyle{apalike}
\bibliography{uniqueness}

\end{document}